\documentclass[a4paper]{llncs}

\usepackage{graphicx}
\usepackage{amssymb}
\usepackage{amsmath}
\usepackage{enumerate}
\usepackage{scalefnt}
\usepackage{setspace} 

\usepackage{url}
\newcommand{\keywords}[1]{\par\addvspace\baselineskip
\noindent\keywordname\enspace\ignorespaces#1}
\newcommand{\bm}[1]{\mbox{\boldmath{$#1$}}}

\newcommand{\argmax}{\operatornamewithlimits{argmax}}
\newcommand{\argmin}{\operatornamewithlimits{argmin}}

\newtheorem{thm}{Theorem}
\newtheorem{lem}{Lemma}
\newtheorem{clm}{Claim}

\begin{document}


\mainmatter

\title{Improved Algorithms for Multiple Sink Location Problems in Dynamic Path Networks}

\author{
Yuya Higashikawa \inst{1}
\and Mordecai J.~Golin \inst{2}
\and Naoki Katoh \inst{1} \thanks{Supported by JSPS Grant-in-Aid for Scientific Research(A)(25240004)}
}

\institute{
Department of Architecture and Architectural Engineering, Kyoto University, Japan, 
\{as.higashikawa, naoki\}@archi.kyoto-u.ac.jp 
\and 
Department of Computer Science and Engineering, The Hong Kong University of Science and Technology, Hong Kong,
golin@cs.ust.hk 
}

\maketitle

\begin{abstract}
This paper considers the $k$-sink location problem in dynamic path networks.
In our model, a dynamic path network consists of an undirected path with positive edge lengths, uniform edge capacity, and positive vertex supplies.
Here, each vertex supply corresponds to a set of evacuees.
Then, the problem requires to find the optimal location of $k$ sinks in a given path so that each evacuee is sent to one of $k$ sinks.
Let ${\bm x}$ denote a $k$-sink location.
Under the optimal evacuation for a given ${\bm x}$, 
there exists a $(k-1)$-dimensional vector ${\bm d}$, called $(k-1)$-divider, 
such that each component represents the boundary dividing all evacuees between adjacent two sinks into two groups, 
i.e., all supplies in one group evacuate to the left sink and all supplies in the other group evacuate to the right sink.
Therefore, the goal is to find ${\bm x}$ and ${\bm d}$ which minimize the maximum cost or the total cost,
which are denoted by the minimax problem and the minisum problem, respectively.
We study the $k$-sink location problem in dynamic path networks with continuous model,
and prove that the minimax problem can be solved in $O(kn)$ time and the minisum problem can be solved in $O(n^2 \cdot \min \{ k, 2^{\sqrt{\log k \log \log n}}\})$ time,
where $n$ is the number of vertices in the given network.
Note that these improve the previous results by \cite{hgk14_2}.
\end{abstract}

\section{Introduction}
The Tohoku-Pacific Ocean Earthquake happened in Japan on March 11, 2011, 
and many people failed to evacuate and lost their lives due to severe attack by tsunamis. 
From the viewpoint of disaster prevention from city planning and evacuation planning,  
it has now become extremely important to establish effective evacuation planning systems against large scale disasters. In particular, 
arrangements of tsunami evacuation buildings in large Japanese cities near the coast has become an urgent issue. 
To determine appropriate tsunami evacuation buildings, we need to consider where evacuation buildings are assigned 
and how to partition a large area into small regions so that one evacuation building is designated in each region. 
This produces several theoretical issues to be considered. 
Among them, this paper focuses on the location problem of multiple evacuation buildings 
assuming that we fix the region such that all evacuees in the region are planned to evacuate to one of these buildings. 
In this paper, we consider the simplest case for which the region consists of a single road.

In order to represent the evacuation, we consider the {\it dynamic} setting in graph networks, which was first introduced by Ford et al.~\cite{ff58}.
In a graph network under the dynamic setting, each vertex is given supply and each edge is given length and capacity which limits the rate of the flow into the edge per unit time.
We call such networks under the dynamic setting {\it dynamic networks}.
Dynamic networks can be considered in discrete and continuous models.
In discrete model, each input value is given as an integer.
Then each supply can be regarded as a set of evacuees, and edge capacity is defined as the maximum number of evacuees who can enter an edge per unit time.
On the other hand, in continuous model, each input value is given as a real number.
Then each supply can be regarded as fluid, and edge capacity is defined as the maximum amount of supply which can enter an edge per unit time.
In either model, we assume that all supply at a vertex is sent to the same sink.
{\it The $k$-sink location problem in dynamic networks} is defined as the problem which requires to find the optimal location of $k$ sinks in a given network 
so that all supply of each vertex is sent to one of $k$ sinks in the shortest time.

For the 1-sink location problem in dynamic networks, the following two criteria can be naturally considered: {\it maximum cost criterion} and {\it total cost criterion}
(in static networks, these criteria correspond to the center problem and the median problem in facility location, respectively).
If a sink location $x$ is given in a dynamic network with discrete model, 
the cost of $x$ for an evacuee is defined as the minimum time required to send him/her to $x$
(by taking into account the congestion).
Then two criteria are defined as the maximum of cost of $x$ for all evacuees and the sum of cost of $x$ for all evacuees, respectively.
Now let us turn to continuous model. 
In continuous model, we define the {\it unit} as the infinitesimally small portion of supply,
then the cost is defined on each unit.
If a sink location $x$ is given in a dynamic network with continuous model, 
the cost of $x$ for a unit is defined as the minimum time required to send the unit to $x$.
Also two criteria are defined as the maximum of cost of $x$ for all units and the sum of cost of $x$ for all units, respectively.
Definitions for $k$-sink location problem are given later.
Then, {\it the minimax} (resp. {\it minisum}) {\it $k$-sink location problem in dynamic networks} requires to find a $k$-sink location in a given dynamic network which minimizes the maximum (resp. total) cost.
Mamada et al.~\cite{mumf06} studied the minimax 1-sink location problem in dynamic tree networks with discrete model assuming that the sink must be located at a vertex,
and proposed an $O(n \log^2 n)$ time algorithm. 
Higashikawa et al.~\cite{hgk14} also studied the same problem as \cite{mumf06} assuming that edge capacity is uniform and the sink can be located at any point in the network,
and proposed an $O(n \log n)$ time algorithm. 
Recently, Higashikawa et al.~\cite{hgk14_2} studied the $k$-sink location problems in a dynamic path network with continuous model assuming that edge capacity is uniform and the sink can be located at any point in the network, and proved that the minimax problem can be solved in $O(kn \log n)$ time and the minisum problem can be solved in $O(kn^2)$ time.

In this paper, we study the same problems as \cite{hgk14_2}, and improve the previous time bounds: 
$O(kn \log n)$ to $O(kn)$ for the minimax problem and $O(kn^2)$ to $O(n^2 \cdot \min \{ k, 2^{\sqrt{\log k \log \log n}}\})$ for the minisum problem.


\section{Minimax $k$-sink location problem}
\label{sec:minimax}

\subsection{Preliminaries}
\label{sec:mmp}

\subsubsection{Model definition:}
Let $P =(V, E)$ be an undirected path where $V = \{ v_1, v_2,$ $\ldots, v_n \}$ and $E = \{ e_1, e_2$, $\ldots, e_{n-1} \}$ 
such that $v_i$ and $v_{i+1}$ are endpoints of $e_i$ for $1 \le i \le n-1$.
Let $\mathcal{N} = (P, l, w, c, \tau)$ be a dynamic network 
with the underlying graph being a path $P$,
$l$ is a function that associates each edge $e_i$ with positive length $l_i$, 
$w$ is also a function that associates each vertex $v_i$ with positive weight $w_i$ representing the amount of supply at $v_i$,
$c$ is a positive constant representing the amount of supply which can enter an edge per unit time,
and $\tau$ is also a constant representing the time required by flow for traversing the unit distance.
We call such networks with path structures {\it dynamic path networks}.
In the following, we use the notation $P$ to denote the set of all points $p \in P$.
Also, for a vertex $v_i \in P$ with $1 \le i \le n$, we abuse the notation $v_i$ to denote the distance from $v_1$ to $v_i$,
and for a point $p \in P$, we abuse the notation $p$ to denote the distance from $v_1$ to $p$.
Then, we can regard $P$ as embedded on a real line such that $v_1 = 0$.
For two points $p, q \in P$ with $p < q$, 
let $[p, q]$ (resp. $[p, q)$, $(p, q]$ and $(p, q)$) denote the part of $P$ which consists of all points $x \in P$ 
such that $p \le x \le q$ (resp. $p \le x < q$, $p < x \le q$ and $p < x < q$).

\subsubsection{$k$-sink location and $(k-1)$-divider:}
Suppose that $k$ sinks are located at points $x_1, x_2, \ldots, x_k \in P$ such that $x_1 \le x_2 \le \ldots \le x_k$, respectively.
Note that each sink can be located at any point in $P$.
In this paper, we assume that if we place a sink at a vertex, all supply of the vertex can finish the evacuation in no time.
So, without loss of generality, we assume $k \le n$
(otherwise, at least one sink can be located at each vertex). 
Let ${\bm x} = (x_1, x_2, \ldots, x_k)$ which is a $k$-dimensional vector, called {\it $k$-sink location}.
Let us consider the optimal evacuation for a given ${\bm x}$.
In this paper, we assume that all units of a vertex are sent to the same sink.
We call a directed path along which all units of a vertex are sent to a sink {\it evacuation path}.
Then, any two evacuation paths never cross each other in an optimal evacuation
(otherwise, we can realize the better or equivalent evacuation by exchanging the two destinations of crossing evacuation paths).
Suppose that there exists only one vertex $v_j$ in $[x_i, x_{i+1}]$ and all units of the vertex are sent to $x_i$,
then $x_{i+1}$ can be moved to $v_{j+1}$ without increasing the cost of any unit.
Therefore, if we optimally locate $k$ sinks with $k \ge 2$, there exist at least two vertices in $[x_i, x_{i+1}]$ for any $i$ with $1 \le i \le k-1$, i.e.,
there exist two vertices $v_j$ and $v_{j+1}$ with $1 \le j \le n-1$ in $[x_i, x_{i+1}]$
such that all supplies on $[x_i, v_j]$ are sent to $x_i$ and
all supplies on $[v_{j+1}, x_{i+1}]$ are sent to $x_{i+1}$.
We call such a vertex $v_j$ {\it dividing vertex}.
For an integer $i$ with $1 \le i \le k-1$ with $k \ge 2$, let $d_i$ be an index of the dividing vertex in $[x_i, x_{i+1})$.
By the above discussion, $d_{i-1}+1 \le d_i$ holds for $1 \le i \le k$ where $d_0 = -1$ and $d_k = n$.
Let ${\bm d} = (d_1, d_2, \ldots, d_{k-1})$ which is a $(k-1)$-dimensional vector, called {\it $(k-1)$-divider}.
For a given ${\bm d}$, we need only consider ${\bm x}$ such that $x_i $ is given on $[v_{d_{i-1}+1}, v_{d_i}]$ for $1 \le i \le k$, where $d_0 = 0$ and $d_k = n$.

\subsubsection{Problem definition:}
For given ${\bm x}$ and ${\bm d}$,
and also for an integer $i$ with $1 \le i \le k$, let $\Theta_i({\bm x}, {\bm d})$ denote the minimum time required to send all supplies on $[v_{d_{i-1}+1}, v_{d_i}]$ to $x_i$, where $d_0 = 0$ and $d_k = n$.
Letting $\Theta({\bm x}, {\bm d}) = \max \{ \Theta_i({\bm x}, {\bm d}) \mid 1 \le i \le k \}$, the minimax $k$-sink location problem is defined as follows:
\begin{eqnarray}
{\rm Q_{minimax}}(P): \ {\rm minimize} \ \left\{ \Theta({\bm x}, {\bm d}) \mid {\bm x} \in P^k \ {\rm and} \ {\bm d} \in \{ 1, 2, \ldots, n \}^{k-1} \right\}.
\label{pro1}
\end{eqnarray}


\subsection{Recursive formulation}
We now consider a subproblem of the above mentioned problem:
for some integers $i, j$ and $p$ with $1 \le i \le j \le n$ and $1 \le p \le k$, the $p$-sink location problem in $[v_i, v_j]$.
For $[v_i, v_j]$, let ${\bm x}^*(p, i, j)$ denote the optimal $p$-sink location and ${\bm d}^*(p, i, j)$ denote the optimal $(p-1)$-divider.
Note that ${\bm x}^*(p, i, j)$ is a $p$-dimensional vector and ${\bm d}^*(p, i, j)$ is also a $(p-1)$-dimensional vector,
so ${\bm d}^*(p, i, j)$ is not defined for $p=1$.
Also, let ${\sf OPT}(p, i, j)$ denote the optimal cost of $p$-sink location in $[v_i, v_j]$, 
i.e., the minimum time required to send all supplies on $[v_i, v_j]$ divided by ${\bm d}^*(p, i, j)$ to ${\bm x}^*(p, i, j)$.
Note that if $p \ge j-i+1$ holds, the optimal sink location is trivial, i.e., ${\sf OPT}(p, i, j) = 0$.

Next, we show the recursive formula of ${\sf OPT}(p, i, j)$.
For integers $i, j$ and $p$ with $1 \le i \le j \le n$ and $1 \le p \le k-1$,
let us consider the optimal $(p+1)$-sink location and $p$-divider for $[v_i, v_j]$, i.e., ${\bm x}^*(p+1, i, j)$ and ${\bm d}^*(p+1, i, j)$.
Since any two evacuation paths never cross each other in an optimal evacuation,
there exists an integer $h$ with $i \le h \le j-1$ such that
all supplies on $[v_{h+1}, x_j]$ are sent to the rightmost sink
and all supplies on $[x_i, v_h]$ are sent to the other $k$ sinks.
Thus, we have the following recursion: 
\begin{eqnarray}
{\sf OPT}(p+1, i, j) = \min_{i \le h \le j-1} \max \{ {\sf OPT}(p, i, h), {\sf OPT}(1, h+1, j) \}.
\label{eq0.0}
\end{eqnarray}
Here, let $d$ be an integer which minimizes the maximum of ${\sf OPT}(p, i, h)$ and ${\sf OPT}(1, h+1, j)$ on $i \le h \le j-1$:
\begin{eqnarray}
d = \argmin_{i \le h \le j-1} \max \{ {\sf OPT}(p, i, h), {\sf OPT}(1, h+1, j) \}.
\label{eq0.1}
\end{eqnarray}
Then, ${\bm x}^*(p+1, i, j)$ and ${\bm d}^*(p+1, i, j)$ can be represented by using $d$ as follows:
\begin{eqnarray}
{\bm x}^*(p+1, i, j) &=& ({\bm x}^*(p, i, d), {\bm x}^*(1, d+1, j)), \label{eq0.2} \\
{\bm d}^*(p+1, i, j) &=& ({\bm d}^*(p, i, d), d). \label{eq0.3}
\end{eqnarray}


\subsection{Known properties of 1-sink location problem}
Here, we introduce the properties of 1-sink location problem, which were explicitly shown in \cite{hgk14_2} (based on \cite{chknsx13,hacgknsx14}). 
For fixed integers $i$ and $j$ with $1 \le i \le j \le n$, let us consider how to compute the optimal 1-sink location in $[v_i, v_j]$.
Suppose that a sink is located at a point $x$ in $[v_i, v_j]$.
Let $\Theta_{i, j}(x)$ denote the minimum time required to send all supplies on $[v_i, v_j]$ to $x$.
Here, let $L_i(x)$ (resp. $R_j(x)$) denote the minimum time required to send all supplies on $[v_i, x]$ (resp. $[x, v_j]$) to $x$ where $L_i(v_i) = 0$ and $R_j(v_j) = 0$.
Then, $\Theta_{i, j}(x)$ is the maximum of $L_i(x)$ and $R_j(x)$,
i.e., 
\begin{eqnarray}
\Theta_{i, j}(x) = \max \{ L_i(x), R_j(x) \}. \label{eq4}
\end{eqnarray}
For discrete model, Kamiyama et al. \cite{kkt06} showed that
$L_i(x)$ and $R_j(x)$ are expressed as follows:
\begin{eqnarray*}
L_i(x) &=& \max_l \left\{ \tau(x - v_l) + \bigg\lceil \frac{\sum_{i \le h \le l} w_h}{c} \bigg\rceil - 1 \ \bigg| \ v_l \in [v_i, x) \right\},  \\
R_j(x) &=& \max_l \left\{ \tau(v_l - x) + \bigg\lceil \frac{\sum_{l \le h \le j} w_h}{c} \bigg\rceil - 1 \ \bigg| \ v_l \in (x, v_j] \right\}.
\end{eqnarray*}
From these, we can immediately develop the formulae for continuous model as follows:
\begin{eqnarray}
L_i(x) &=& \max_l \left\{ \tau(x - v_l) + \frac{\sum_{i \le h \le l} w_h}{c} \ \bigg| \ v_l \in [v_i, x) \right\}, \label{eq3.1} \\
R_j(x) &=& \max_l \left\{ \tau(v_l - x) + \frac{\sum_{l \le h \le j} w_h}{c} \ \bigg| \ v_l \in (x, v_j] \right\}. \label{eq3.2}
\end{eqnarray}
Note that $L_i(x)$ (resp. $R_j(x)$) is a piecewise linear strictly increasing (resp. decreasing) function of $x$.
Therefore, a function $\Theta_{i, j}(x)$ is unimodal in $x$, and there exists the unique point which minimizes $\Theta_{i, j}(x)$, they is, ${\bm x}^*(1, i, j)$.
Then, as \cite{chknsx13,hacgknsx14,hgk14_2} showed, we immediately have the following claim.
\begin{clm}
For any integers $i$ and $j$ with $1 \le i \le j \le n$ and a point $x \in [v_i, v_j]$, \\
{\rm (i)} if $L_i(x) \le R_j(x)$ holds, ${\bm x}^*(1, i, j) \ge x$ holds, and \\
{\rm (ii)} if $L_i(x) \ge R_j(x)$ holds, ${\bm x}^*(1, i, j) \le x$ holds.
\label{clm2}
\end{clm}
In the following, when $x$ is at a vertex $v_t$ with $i \le t \le j$,
we use the notation $L(i, t)$ (resp. $R(t, j)$) to denote the value $L_i(v_t)$ (resp. $R_j(v_t)$).
Then, we have the following claim (which was also shown in \cite{chknsx13,hacgknsx14,hgk14_2}).
\begin{clm}
For given integers $i$ and $j$ with $1 \le i \le j \le n$, 
suppose that for the interval $[v_l, v_{l+1}]$ with $i \le l \le j-1$, $L(i, l) \le R(l, j)$ and $L(i, l+1) \ge R(l+1, j)$ hold,
and let $\alpha^*$ denote the solution to an equation for $\alpha$: $R(l, j) - \alpha\tau(v_{l+1} - v_l) = L(i, l+1) - (1 - \alpha)\tau(v_{l+1} - v_l)$.
Then, \\
{\rm (i)} if $1 \le \alpha^* \le 1$ holds, ${\bm x}^*(1, i, j)$ is a point dividing the interval $[v_l, v_{l+1}]$ with the ratio of $\alpha^*$ to $1-\alpha^*$ 
and ${\sf OPT}(1, i, j) = R(l, j) - \alpha^*\tau(v_{l+1} - v_l)$ holds, \\
{\rm (ii)} if $\alpha^* < 0$ holds, ${\bm x}^*(1, i, j) = v_l$ and ${\sf OPT}(1, i, j) = R(l, j)$ hold, and \\
{\rm (iii)} if $\alpha^* > 1$ holds, ${\bm x}^*(1, i, j) = v_{l+1}$ and ${\sf OPT}(1, i, j) = L(i, l+1)$ hold. \\
\label{clm2.1}
\end{clm}


\subsection{Key properties of $k$-sink location problem}
\label{sec:mmkp}
In this section, we show several key properties of the $k$-sink location problem.
Here, 
for integers $p$ and $i$ with $2 \le p \le k$ and $2 \le i \le n$, 
let $f_{p, i}(t)$ denote a function defined on $\{ t \in \mathbb{Z} \mid 1 \le t \le i-1 \}$:
\begin{eqnarray}
f_{p, i}(t) = \max \{ {\sf OPT}(p-1, 1, t), {\sf OPT}(1, t+1, i) \}.
\label{eq1}
\end{eqnarray}
Note that for fixed $p$ and $i$, ${\sf OPT}(p-1, 1, t)$ is monotonically increasing in $t$ and ${\sf OPT}(1, t+1, i)$ is monotonically decreasing in $t$.
Thus, we have the following claim.
\begin{clm}
For any integers $p$ and $i$ with $2 \le p \le k$ and $2 \le i \le n$,
function $f_{p, i}(t)$ is unimodal in $t$ on $1 \le t \le i-1$.
\label{clm1}
\end{clm}
Let $d_{p, i}$ be an integer which minimizes $f_{p, i}(t)$ for $1 \le t \le i-1$:
\begin{eqnarray}
d_{p, i} = \argmin_{1 \le t \le i-1} f_{p, i}(t).
\label{eq1.1}
\end{eqnarray}
By Claim \ref{clm1}, there uniquely exists $d_{p, i}$.
By (\ref{eq0.2}) and (\ref{eq0.3}), we have
\begin{eqnarray}
{\bm x}^*(p, 1, i) &=& ({\bm x}^*(p-1, 1, d_{p, i}), {\bm x}^*(1, d_{p, i} + 1, i)), \label{eq2.1} \\
{\bm d}^*(p, 1, i) &=& ({\bm d}^*(p-1, 1, d_{p, i}), d_{p, i}). \label{eq2.2}
\end{eqnarray}
Then, we prove the following two lemmas.
\begin{lem}
For any integers $p$ and $i$ with $2 \le p \le k$ and $2 \le i \le n-1$,
$d_{p, i} \le d_{p, i+1}$ holds.
\label{lem1}
\end{lem}

\begin{lem}
For any integers $h, i, j$ and $l$ with $1 \le i \le j \le n$, $1 \le h \le l \le n$, $i \le h$ and $j \le l$,
${\bm x}^*(1, i, j) \le {\bm x}^*(1, h, l)$ holds.
\label{lem2}
\end{lem}

\subsubsection{Proof of Lemma \ref{lem1}:}
In order to prove Lemma \ref{lem1}, we first confirm a fundamental property.
\begin{clm}
For any integers $p$ with $1 \le p \le k$, and $h, i, j$ and $l$ with $1 \le h \le i \le j \le l \le n$,
${\sf OPT}(p, i, j) \le {\sf OPT}(p, h, l)$ holds.
\label{clm:mm1.1}
\end{clm}
We prove Lemma \ref{lem1} by contradiction: 
there exist integers $p$ and $i$ with $2 \le p \le k$ and $2 \le i \le n-1$ such that
$d_{p, i} > d_{p, i+1}$ holds.
For ease of notation in the proof, we use the notations $A, B, C, D, E$ and $F$ as follows:
\begin{eqnarray}
\begin{array}{ll}
A = {\sf OPT}(p-1, 1, d_{p, i}),		&	B = {\sf OPT}(1, d_{p, i}+1, i), \\
C = {\sf OPT}(p-1, 1, d_{p, i+1}),	&	D = {\sf OPT}(1, d_{p, i+1}+1, i+1), \\
E = {\sf OPT}(1, d_{p, i+1}+1, i),	&	F = {\sf OPT}(1, d_{p, i}+1, i+1).
\end{array}
\label{eq:mm1.1}
\end{eqnarray}
From the assumption of $d_{p, i} > d_{p, i+1}$ and Claim \ref{clm:mm1.1}, we can derive the following inequalities:
\begin{eqnarray}
C &\le& A, \label{eq:mm1.5} \\
B &\le& E \le D, \label{eq:mm1.6} \\
B &\le& F \le D. \label{eq:mm1.7}
\end{eqnarray}
Since $d_{p, i}$ minimizes $f_{p, i}(t) = \max \{ {\sf OPT}(p-1, 1, t), {\sf OPT}(1, t+1, i) \}$ (refer to (\ref{eq1}) and (\ref{eq1.1})),
we have the following inequality:
\begin{eqnarray}
\max \{A, B \} \le \max \{ C, E \}.
\label{eq:mm1.2}
\end{eqnarray}
Also, without loss of generality, we assume that $d_{p, i+1}$ is maximized unless the cost increases. 
By this assumption, we have the following inequality:
\begin{eqnarray}
\max \{C, D \} < \max \{ A, F \}.
\label{eq:mm1.3}
\end{eqnarray}
Then, we consider three cases:
[Case 1] $A \le B$; [Case 2] $D \le C$; [Case 3] $B < A$ and $C < D$. \\
\noindent
[Case 1]: By (\ref{eq:mm1.5}), (\ref{eq:mm1.7}) and the condition of $A \le B$, we have $C \le A \le F \le D$, which contradicts (\ref{eq:mm1.3}).\\
\noindent
[Case 2]: By (\ref{eq:mm1.5}), (\ref{eq:mm1.6}) and the condition of $D \le C$, we have $B \le E \le C \le A$.
By this and (\ref{eq:mm1.2}), we have $A \le C$.
Also, by (\ref{eq:mm1.5}), (\ref{eq:mm1.7}) and the condition of $D \le C$, we have $F \le D \le C \le A$.
By this and (\ref{eq:mm1.3}), we have $C < A$, which contradicts $A \le C$.\\
\noindent
[Case 3]: By (\ref{eq:mm1.2}) and the condition of $B < A$, we have 
\begin{eqnarray}
A \le \max \{ C, E \}. \label{eq:mm1.8}
\end{eqnarray}
Also, by (\ref{eq:mm1.3}) and the condition of $C < D$, we have
\begin{eqnarray}
D < \max \{ A, F \}. \label{eq:mm1.9}
\end{eqnarray}
If $F \le A$ holds, we have $D < \max \{ C, E \}$ by (\ref{eq:mm1.8}) and (\ref{eq:mm1.9}), which contradicts the condition of $C < D$ or (\ref{eq:mm1.6}).
If $A < F$ holds, we have $D < F$ by (\ref{eq:mm1.9}), which contradicts (\ref{eq:mm1.7}).
\qed

\subsubsection{Proof of Lemma \ref{lem2}:}
In order to prove Lemma \ref{lem2}, we first confirm the following claim (refer to the definitions of (\ref{eq3.1}) and (\ref{eq3.2})).
\begin{clm}
{\rm (i)} For any integers $i$ and $j$ with $1 \le j \le i \le n$ and any points $x$ and $y$ with $v_i \le x \le y \le v_n$,
$L_i(x) \le L_j(x)$ and $L_i(x) \le L_i(y)$ hold.\\
{\rm (ii)} For any integers $i$ and $j$ with $1 \le i \le j \le n$ and any points $x$ and $y$ with $v_1 \le y \le x \le v_i$,
$R_i(x) \le R_j(x)$ and $R_i(x) \le R_i(y)$ hold.
\label{clm:mm2.1}
\end{clm}
We prove Lemma \ref{lem2} by contradiction: 
there exist integers $h, i, j$ and $l$ with $1 \le i \le j \le n$, $1 \le h \le l \le n$, $i \le h$ and $j \le l$ such that
${\bm x}^*(1, i, j) > {\bm x}^*(1, h, l)$ holds.
By this assumption, we have the following inequality:
\begin{eqnarray}
i \le h \le {\bm x}^*(1, h, l) < {\bm x}^*(1, i, j) \le j \le l.
\label{eq:mm2.4}
\end{eqnarray}
For ease of notation in the proof, we use the notations $A, B, C, D, E, F, G$ and $H$ as follows:
\begin{eqnarray}
\begin{array}{ll}
A = L_i({\bm x}^*(1, i, j)),		&	B = R_j({\bm x}^*(1, i, j)), \\
C = L_h({\bm x}^*(1, h, l)),		&	D = R_l({\bm x}^*(1, h, l)), \\
E = L_i({\bm x}^*(1, h, l)),		&	F = R_j({\bm x}^*(1, h, l)), \\
G = L_h({\bm x}^*(1, i, j)),		&	H = R_l({\bm x}^*(1, i, j)).
\end{array}
\label{eq:mm2.1}
\end{eqnarray}
From (\ref{eq:mm2.4}) and Claim \ref{clm:mm2.1}, we can derive the following inequalities:
\begin{eqnarray}
C &\le& E \le A, \label{eq:mm2.5} \\
C &\le& G \le A, \label{eq:mm2.6} \\
B &\le& F \le D, \label{eq:mm2.7} \\
B &\le& H \le D. \label{eq:mm2.8}
\end{eqnarray}
Since ${\bm x}^*(1, i, j)$ and ${\bm x}^*(1, h, l)$ are the unique points which minimize $\Theta_{i, j}(x) = \max \{ L_i(x),$ $R_j(x)\}$ and $\Theta_{h, l}(x) = \max \{ L_h(x), R_l(x)\}$, respectively
(refer to (\ref{eq4})), we have the following inequalities:
\begin{eqnarray}
\max \{A, B \} < \max \{ E, F \}, \label{eq:mm2.2} \\
\max \{C, D \} < \max \{ G, H \}. \label{eq:mm2.3}
\end{eqnarray}
Then, we consider three cases:
[Case 1] $A \le B$; [Case 2] $D \le C$; [Case 3] $B < A$ and $C < D$. \\
\noindent
[Case 1]: By (\ref{eq:mm2.6}), (\ref{eq:mm2.8}) and the condition of $A \le B$, we have $C \le G \le H \le D$, which contradicts (\ref{eq:mm2.3}).\\
\noindent
[Case 2]: By (\ref{eq:mm2.5}), (\ref{eq:mm2.7}) and the condition of $D \le C$, we have $B \le F \le E \le A$, which contradicts (\ref{eq:mm2.2}).\\
\noindent
[Case 3]: By (\ref{eq:mm2.2}) and the condition of $B < A$, we have 
\begin{eqnarray}
A < \max \{ E, F \}. \label{eq:mm2.9}
\end{eqnarray}
Also, by (\ref{eq:mm2.3}) and the condition of $C < D$, we have
\begin{eqnarray}
D < \max \{ G, H \}. \label{eq:mm2.10}
\end{eqnarray}
If $F \le E$ holds, we have $A < E$ by (\ref{eq:mm2.9}), which contradicts (\ref{eq:mm2.5}).
Also, if $G \le H$ holds, we have $D < H$ by (\ref{eq:mm2.10}), which contradicts (\ref{eq:mm2.8}).
If $E < F$ and $H < G$ hold, we have $A < F \le D < G$ by (\ref{eq:mm2.7}), (\ref{eq:mm2.9}) and (\ref{eq:mm2.10}), 
that is, $A < G$ holds, which contradicts (\ref{eq:mm2.6}).
\qed


\subsection{Algorithm based on dynamic programming}
\label{sec:mma}
The algorithm basically computes ${\sf OPT}(1, 1, 1)$, $\ldots$, ${\sf OPT}(1, 1, n)$, ${\sf OPT}(2, 1, 1)$, $\ldots$, ${\sf OPT}(2, 1, n)$, $\ldots$, ${\sf OPT}(k, 1, 1)$, $\ldots$, ${\sf OPT}(k, 1, n)$ in this order.
For some integers $p$ and $i$ with $2 \le p \le k$ and $2 \le i \le n$,
let us consider how to obtain ${\sf OPT}(p, 1, i)$.
Actually, in order to obtain ${\sf OPT}(p, 1, i)$, the algorithm needs ${\sf OPT}(p-1, 1, l)$ for $l = 1, 2, \ldots, n$ and ${\sf OPT}(p, 1, i-1)$, which are supposed to have been obtained.
By (\ref{eq0.0}), (\ref{eq1}) and (\ref{eq1.1}), we have
\begin{eqnarray}
{\sf OPT}(p, 1, i) = f_{p, i}(d_{p, i}) = \max \{ {\sf OPT}(p-1, 1, d_{p, i}), {\sf OPT}(1, d_{p, i}+1, i) \}. \label{eq8}
\end{eqnarray}
Here, we assumed that ${\sf OPT}(p-1, 1, d_{p, i})$ has already been obtained.
Thus, in order to obtain ${\sf OPT}(p, 1, i)$, we only need to compute ${\sf OPT}(1, d_{p, i} + 1, i)$.
Recall that $d_{p, i}$ is the unique point which minimizes function $f_{p, i}(t)$ (refer to (\ref{eq1}) and (\ref{eq1.1})).
Now, the algorithm knows where $d_{p, i-1}$ exists, and by Lemma \ref{lem1}, $d_{p, i-1} \le d_{p, i}$ holds.
So the algorithm starts to compute $f_{p, i}(t)$ for $t=d_{p, i-1}$,
and continues to compute in ascending order of $t$, as will be shown below.
Note that function $f_{p, i}(t)$ is unimodal in $t$ by Claim \ref{clm1}, 
which implies that $f_{p, i}(t)$ is strictly decreasing until $t=d_{p, i}$.
Thus, if the algorithm reaches the first integer $t^* \ge d_{p, i-1}$ such that $f_{p, i}(t^*) \le f_{p, i}(t^*+1)$,
it outputs $t^*$ as $d_{p, i}$.
Then, the algorithm also outputs $f_{p, i}(t^*)$ as ${\sf OPT}(p, 1, i)$.

\subsubsection{Computation of $f_{p, i}(t)$ for $t = d_{p, i-1}$:}
As above mentioned, the algorithm first computes $f_{p, i}(t)$ with $t = d_{p, i-1}$ which is defined as follows:
\begin{eqnarray}
f_{p, i}(d_{p, i-1}) = 	\max \{ {\sf OPT}(p-1, 1, d_{p, i-1}), {\sf OPT}(1, d_{p, i-1}+1, i) \}. \label{eq11}
\end{eqnarray}
Since the algorithm has already obtained ${\sf OPT}(p-1, 1, d_{p, i-1})$,
we only need to compute ${\sf OPT}(1, d_{p, i-1}+1, i)$.
To do this, we actually need to find ${\bm x}^*(1, d_{p, i-1}+1, i)$.
On the other hand, the algorithm has already obtained ${\sf OPT}(p, 1, i-1)$ as follows:
\begin{eqnarray}
{\sf OPT}(p, 1, i-1) = 	\max \{ {\sf OPT}(p-1, 1, d_{p, i-1}), {\sf OPT}(1, d_{p, i-1}+1, i-1) \}, \label{eq12}
\end{eqnarray}
which implies that ${\bm x}^*(1, d_{p, i-1}+1, i-1)$ has been obtained.
By Lemma \ref{lem2}, ${\bm x}^*(1, d_{p, i-1}+1, i-1) \le {\bm x}^*(1, d_{p, i-1}+1, i)$ holds.
Let $l$ and $l'$ be the indices of vertices  
such that ${\bm x}^*(1, d_{p, i-1}+1, i-1) \in [v_{l}, v_{l+1}]$ with $d_{p, i-1}+1 \le l \le i-2$
and ${\bm x}^*(1, d_{p, i-1}+1, i) \in [v_{l'}, v_{l'+1}]$ with $d_{p, i-1}+1 \le l' \le i-1$, respectively (see Figure \ref{fig1}).
\begin{figure}[h]
\centering
\includegraphics[width=70mm,clip]{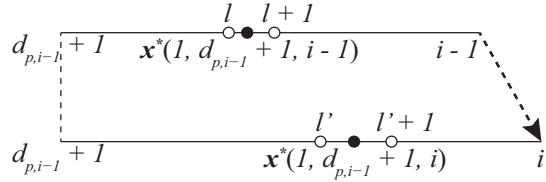} 
\caption{Illustrations of ${\bm x}^*(1, d_{p, i-1}+1, i-1)$ and ${\bm x}^*(1, d_{p, i-1}+1, i)$}
\label{fig1}
\end{figure}
By Claim \ref{clm2}, for any interval $[v_h, v_{h+1}]$ with $d_{p, i-1}+1 \le h \le i-1$, there exists ${\bm x}^*(1, d_{p, i-1}+1, i)$ in $[v_h, v_{h+1}]$
if $L(d_{p, i-1}+1, h) \ge R(h, i)$ and $L(d_{p, i-1}+1, h+1) \le R(h+1, i)$ hold.
Therefore, if we maintain the data structure so that we can compute these values,
the algorithm can test if there exists ${\bm x}^*(1, d_{p, i-1}+1, i) \in [v_h, v_{h+1}]$ or not
(what the data structure is or how we can maintain and use it will be explained in the next subsection).
Then, the algorithm starts to test for $h=l$,
and continues to test in ascending order of $h$.
If an interval where ${\bm x}^*(1, d_{p, i-1}+1, i)$ exists, that is, $[v_{l'}, v_{l'+1}]$ is found,
then ${\bm x}^*(1, d_{p, i-1}+1, i)$ and ${\sf OPT}(1, d_{p, i-1}+1, i)$ can be computed in $O(1)$ time by Claim \ref{clm2.1}.

\subsubsection{Computation of $f_{p, i}(t)$ for $t \ge d_{p, i-1}+1$:}
Now, suppose that for an integer $t$ with $t \ge d_{p, i-1}$,
the algorithm has already obtained $f_{p, i}(t)$, that is, ${\bm x}^*(1, t+1, i)$ and ${\sf OPT}(1, t+1, i)$.
For an integer $t$ with $t \ge d_{p, i-1}$, let $l(t+1)$ be the index of a vertex with $t+1 \le l(t+1) \le i-1$
such that ${\bm x}^*(1, t+1, i) \in [v_{l(t+1)}, v_{l(t+1)+1}]$.
Note that $l(t+1)$ has also been obtained (see Figure \ref{fig2}).
Then, the computation of $f_{p, i}(t+1)$ comes down to finding $l(t+2)$ which is greater than or equal to $l(t+1)$,
and so, it can be treated in the similar manner as the computation of $f_{p, i}(d_{p, i-1})$.
\begin{figure}[h]
\centering
\includegraphics[width=70mm,clip]{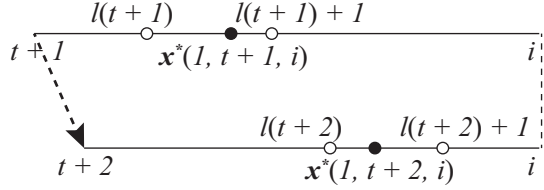} 
\caption{Illustrations of ${\bm x}^*(1, t+1, i)$ and ${\bm x}^*(1, t+2, i)$}
\label{fig2}
\end{figure}

\subsection{How to compute $L(\alpha, \beta)$ and $R(\beta, \gamma)$}
\label{s2}
As mentioned in Section \ref{sec:mma}, in order to obtain ${\sf OPT}(p, 1, i)$ for fixed $p$ and all $i=p+1, p+2, \ldots, n$
(note that ${\sf OPT}(p, 1, i)$ $ = 0$ for $i=1, 2, \ldots, p$), 
the algorithm computes $f_{p, p+1}(d_{p, p}),\ldots,f_{p, p+1}(d_{p, p+1})$,
$f_{p, p+2}(d_{p, p+1}),\ldots,f_{p, p+2}(d_{p, p+2}),\ldots,$
$f_{p, n}(d_{p, n-1}),\ldots,f_{p, n}(d_{p, n})$.
In this computation,
the algorithm actually computes $L(p, p),L(p, p+1),\ldots,L(p, l(p)),L(p+1, l(p)),L(p+1, l(p)+1),\ldots, L(p+1, l(p+1)),\ldots$
where $l(i)$ is the index of vertex with $p \le l(i) \le l(i+1) \le n$ for any $i \ge p$,
and also, $R(p, p),R(p, p+1),\ldots,R(p, r(p)),R(p+1, r(p)),R(p+1, r(p)+1),\ldots, R(p+1, r(p+1)),\ldots$
where $r(i)$ is the index of vertex with $p \le r(i) \le r(i+1) \le n$ for any $i \ge p$.
In order to compute $L(\alpha, \beta)$ and $R(\beta, \gamma)$ for any integers $\alpha, \beta$ and $\gamma$ with $1 \le \alpha \le \beta \le \gamma \le n$, 
the algorithm maintains the specific data structures $D_L(\alpha, \beta)$ and $D_R(\beta, \gamma)$, respectively.
Depending on the situation, the algorithm updates $D_L(\alpha, \beta)$ to $D_L(\alpha+1, \beta)$ or $D_L(\alpha, \beta+1)$ and $D_R(\beta, \gamma)$ to $D_R(\beta+1, \gamma)$ or $D_R(\beta, \gamma+1)$.
We below show definitions of the two data structures and how to maintain these.

\subsubsection{Definition of $D_L(\alpha, \beta)$ and how to maintain $D_L(\alpha, \beta)$:} 
In this discussion, we assume that $\alpha < \beta$ holds ($D_L(\alpha, \beta) = \emptyset$ when $\alpha = \beta$).
Let us consider the evacuation of all supplies on $[v_\alpha, v_\beta]$ to $v_\beta$.
We define the vertex indices $\rho_1, \ldots, \rho_e$ as 
\begin{eqnarray}
\begin{array}{ll}
	\rho_1 = \argmax \left\{ \tau(v_\beta - v_j) + \frac{\sum_{l=\alpha}^j w_l}{c} \ \bigg| \ \alpha \le j < \beta \right\}		 			& \mbox{and} \\
	\rho_i = \argmax \left\{ \tau(v_\beta - v_j) + \frac{\sum_{l=\rho_{i-1}+1}^j w_l}{c} \ \bigg| \ \rho_{i-1} < j < \beta \right\}	& \mbox{for} \ 2 \le i \le e.
    \end{array}
 \label{eq13}
\end{eqnarray}
Note that $\rho_e = \beta-1$ holds. 
For every integer $i$ with $1 \le i \le e$, we also define the value of $\rho_i$ as $\sigma_i = \sum \{ w_h \mid \rho_{i-1}+1 \le h \le \rho_i \}$ where $\rho_0+1 = \alpha$.
Here, we notice that for every integer $i$ with $2 \le i \le e$, the first unit of $v_{\rho_{i-1}}$ never be induced to stop at any vertex $v_j$ with $\rho_{i-1} < j < \beta$.
The data structure $D_L(\alpha, \beta)$ consists of the two sequences $(\rho_1, \ldots, \rho_e)$ and $(\sigma_1, \ldots, \sigma_e)$.
Note that we define the size of $D_L(\alpha, \beta)$ as $|D_L(\alpha, \beta)| = e$.
Recall that in continuous model, the cost is defined on each infinitesimal unit of supply,
i.e., the cost of $x$ for a unit is defined as the minimum time required to send the unit to $x$.
Here, we notice that for any integer $i$ with $2 \le i \le e$, the first unit of $v_{\rho_{i-1}}$ never be induced to stop at $v_{\rho_i}$.
Then, $L(\alpha, \beta)$ can be computed as 
\begin{eqnarray}
L(\alpha, \beta) =  \tau(v_\beta - v_{\rho_1}) + \frac{\sigma_1}{c}.
\label{eq14}
\end{eqnarray}

In order to update $D_L(\alpha, \beta)$ to $D_L(\alpha+1, \beta)$, the algorithm tests if $\rho_1 = \alpha$ holds or not. 
If it holds, the algorithm sets $D_L(\alpha+1, \beta)$ so that
\begin{eqnarray}
    \begin{array}{lllll}
	\rho_i \leftarrow \rho_{i+1} \	& \mbox{and} \ & \sigma_i \leftarrow \sigma_{i+1} \				& \mbox{for} \ & 1 \le i \le e-1.
    \end{array}
    \label{eq15}
\end{eqnarray}
Otherwise, the algorithm sets $D_L(\alpha+1, \beta)$ so that
\begin{eqnarray}
    \begin{array}{lllll}
    	\rho_1 \leftarrow \rho_1 \	& \mbox{and} \ & \sigma_1 \leftarrow \sigma_1 - w_\alpha, \		& \\
	\rho_i \leftarrow \rho_i \	& \mbox{and} \ & \sigma_i \leftarrow \sigma_i \				& \mbox{for} \ & 2 \le i \le e.
    \end{array}
    \label{eq16}
\end{eqnarray}

On the other hand, in order to update $D_L(\alpha, \beta)$ to $D_L(\alpha, \beta+1)$, the algorithm first sets $\rho_{e+1} = \beta$ and $\sigma_{e+1} = w_\beta$.
Then, the algorithm repeatedly tests if $\tau(v_{\rho_{j+1}}-v_{\rho_j}) \le \sigma_{j+1}/c$ holds or not in descending order of $j$ from $j = e$.
If it holds, the algorithm sets $D_L(\alpha, \beta+1)$ so that
\begin{eqnarray}
    \begin{array}{lllll}
	\rho_i \leftarrow \rho_i \		& \mbox{and} \ & \sigma_i \leftarrow \sigma_i \					& \mbox{for} \ & 1 \le i \le j-1, \\
	\rho_j \leftarrow \rho_{j+1} \	& \mbox{and} \ & \sigma_j \leftarrow \sigma_j + \sigma_{j+1}, \	& 	
    \end{array}
    \label{eq17}
\end{eqnarray}
until $\tau(v_{\rho_{e'+1}}-v_{\rho_{e'}}) > \sigma_{e'+1}/c$ holds for $j = e'$ with some integer $e' \le e$.
Let $t(\alpha, \beta)$ denote the number of such tests required to update $D_L(\alpha, \beta)$ to $D_L(\alpha, \beta+1)$, which can be represent as
\begin{eqnarray}
t(\alpha, \beta) = e - e' + 1 = |D_L(\alpha, \beta)| - |D_L(\alpha, \beta+1)| + 2.
\label{eq18}
\end{eqnarray}

Recall that in the computation to obtain ${\sf OPT}(p, 1, i)$ for fixed $p$ and all $i=p+1, p+2, \ldots, n$
for a given integer $\alpha$ with $p \le \alpha \le n-1$, the algorithm updates $D_L(\alpha, l(\alpha-1))$ to $D_L(\alpha, l(\alpha))$ where $l(p-1)=p$ and $l(n) = n$.
Let $T(\alpha)$ denote the total number of such tests required to update $D_L(\alpha, l(\alpha-1))$ to $D_L(\alpha, l(\alpha))$,
and $T$ denote the sum of $T(\alpha)$ for $p \le \alpha \le n-1$.
By (\ref{eq15}) and (\ref{eq16}), we have $|D_L(\alpha, l(\alpha))| \ge |D_L(\alpha+1, l(\alpha))|$, so the upper bound of $T$ can be obtained as
\begin{eqnarray}
T 	&=&		\sum_{\alpha=p}^{n-1} T(\alpha)	=	\sum_{\alpha=p}^{n-1} \sum_{\beta=l(\alpha-1)}^{l(\alpha)} t(\alpha, \beta) \nonumber \\
 	&=&		\sum_{\alpha=p}^{n-1} \big\{ |D_L(\alpha, l(\alpha-1))| - |D_L(\alpha, l(\alpha))| + 2l(\alpha) - 2l(\alpha-1) \big\} \nonumber \\
	&\le&	\sum_{\alpha=p}^{n-1} \big\{ |D_L(\alpha, l(\alpha-1))| - |D_L(\alpha+1, l(\alpha))| + 2l(\alpha) - 2l(\alpha-1) \big\} \nonumber \\
	&=&		|D_L(p, l(p-1))| + 2l(n-1) - 2l(p-1) \in O(n-p),
\label{eq19}
\end{eqnarray}
which implies that $t(\alpha, \beta)$ is amortized $O(1)$.

\subsubsection{Definition of $D_R(\beta, \gamma)$ and how to maintain $D_R(\beta, \gamma)$:} 
In this discussion, we assume that $\beta < \gamma$ holds ($D_R(\beta, \gamma) = \emptyset$ when $\beta = \gamma$).
Let us consider the evacuation of all supplies on $[v_\beta, v_\gamma]$ to $v_\beta$.
We define the vertex indices $\mu_1, \ldots, \mu_f$ as
\begin{eqnarray}
\begin{array}{ll}
	\mu_1 = \argmax \left\{ \tau(v_j - v_\beta) + \frac{\sum_{l=j}^\gamma w_l}{c} \ \bigg| \ \beta < j \le \gamma \right\}			 	& \mbox{and} \\
	\mu_i = \argmax \left\{ \tau(v_j - \mu_{i-1}) + \frac{\sum_{l=j}^\gamma w_l}{c} \ \bigg| \ \mu_{i-1} < j \le \gamma \right\}	& \mbox{for} \ 2 \le i \le f.
    \end{array}
 \label{eq20}
\end{eqnarray}
Note that $\mu_f = \gamma$ holds. 
For every integer $i$ with $1 \le i \le f$, we also define the value of $\mu_i$ as $W_i = \sum \{ w_h \mid \mu_i \le h \le n\}$.
Here, we notice that $v_{\mu_i}$ is the rightmost vertex of which the first unit never be induced to stop at any vertex $v_j$ with $\mu_{i-1} < j < \mu_i$ where $\mu_0 = \beta$.
In addition, let $os(\gamma) = \sum \{ w_h \mid \gamma+1 \le h \le n\}$.
The data structure $D_R(\beta, \gamma)$ consists of the offset value $os(\gamma)$ and the two sequences $(\mu_1, \ldots, \mu_f)$ and $(W_1, \ldots, W_f)$.
Note that we define the size of $D_R(\beta, \gamma)$ as $|D_R(\beta, \gamma)| = f$.
Then, $R(\beta, \gamma)$ can be computed as 
\begin{eqnarray}
R(\beta, \gamma) =  \tau(v_{\mu_1} - v_\beta) + \frac{W_1 - os(\gamma)}{c}.
\label{eq21}
\end{eqnarray}

In order to update $D_R(\beta, \gamma)$ to $D_R(\beta+1, \gamma)$, the algorithm tests if $\mu_1 = \beta+1$ holds or not. 
If it holds, the algorithm sets $D_R(\beta+1, \gamma)$ so that
\begin{eqnarray}
    \begin{array}{lllll}
	\mu_i \leftarrow \mu_{i+1} \	& \mbox{and} \ & W_i \leftarrow W_{i+1} \				& \mbox{for} \ & 1 \le i \le f-1.
    \end{array}
    \label{eq22}
\end{eqnarray}
Otherwise, nothing changes, that is, the algorithm sets $D_R(\beta+1, \gamma) = D_R(\beta, \gamma)$.

On the other hand, in order to update $D_R(\beta, \gamma)$ to $D_R(\beta, \gamma+1)$, the algorithm first sets $\mu_{f+1} = \gamma+1$ and 
compute $W_{f+1} = W_f - w_\gamma$ and $os(\gamma+1) = os(\gamma) - w_{\gamma+1}$.
Then, the algorithm repeatedly tests if 
$\tau v_{\gamma+1} + w_{\gamma+1}/c \ge \tau v_{\mu_j} + (W_j - os(\gamma+1))/c$ holds or not in descending order of $j$ from $j = f$.
If it holds, the algorithm sets $D_R(\beta, \gamma+1)$ so that
\begin{eqnarray}
    \begin{array}{lllll}
	\mu_i \leftarrow \mu_i \		& \mbox{and} \ & W_i \leftarrow W_i \					& \mbox{for} \ & 1 \le i \le j-1, \\
	\mu_{j} \leftarrow \mu_{j+1} \	& \mbox{and} \ & W_{j} \leftarrow W_{j+1}, \	& 	
    \end{array}
    \label{eq23}
\end{eqnarray}
until $\tau v_{\gamma+1} + w_{\gamma+1}/c < \tau v_{\mu_{f'}} + (W_{f'} - os(\gamma+1))/c$ holds for $j=f'$ with some integer $f' \le f$.
Let $t'(\beta, \gamma)$ denote the number of such tests required to update $D_R(\beta, \gamma)$ to $D_R(\beta, \gamma+1)$, which can be represent as
\begin{eqnarray}
t'(\beta, \gamma) = f - f' + 1 = |D_R(\beta, \gamma)| - |D_R(\beta, \gamma+1)| + 2,
\label{eq24}
\end{eqnarray}
which is amortized $O(1)$ by the same discussion as that for $t(\alpha, \beta)$ defined at (\ref{eq18}).

\begin{clm}
For any integers $\alpha, \beta$ and $\gamma$ with $1 \le \alpha \le \beta \le \gamma \le n$,
$L(\alpha, \beta)$ and $R(\beta, \gamma)$ can be computed in $O(1)$ time once $D_L(\alpha, \beta)$ and $D_R(\beta, \gamma)$ have been obtained.
\label{clm:ds1}
\end{clm}
\begin{clm}
{\rm (i)} For any integers $\alpha, \beta$ and $\gamma$ with $1 \le \alpha < \beta < r \le n$,
$L(\alpha, \beta)$ and $R(\beta, \gamma)$ can be updated to $L(\alpha+1, \beta)$ and $R(\beta+1, \gamma)$ in amortized $O(1)$ time, respectively. \\
{\rm (ii)} For any integers $\alpha, \beta$ and $\gamma$ with $1 \le \alpha \le \beta \le \gamma \le n-1$,
$L(\alpha, \beta)$ and $R(\beta, \gamma)$ can be updated to $L(\alpha, \beta+1)$ and $R(\beta, \gamma+1)$ in amortized $O(1)$ time, respectively.
\label{clm:ds2}
\end{clm}

\subsection{Time complexity}

As mentioned in Section \ref{sec:mma} and at the beginning of Section \ref{s2}, in order to obtain ${\sf OPT}(p, 1, i)$ for fixed $p$ and all $i=p+1, p+2, \ldots, n$, 
$O(n)$ intervals are tested in total as follows:
in order to test if there exists ${\bm x}^*(1, i, j)$ in an interval $[v_h, v_{h+1}]$ or not, 
the algorithm needs to confirm that $L(i, h) \ge R(h, j)$ and $L(i, h+1) \le R(h+1, j)$ hold by Claim \ref{clm2}, 
which takes $O(1)$ time once $D_L(i, h), D_L(i, h+1), D_R(h, j)$ and $D_R(h+1, j)$ have been obtained by Claim \ref{clm:ds1}.
Thus, such computations take $O(n)$ time in total.

On the other hand, let us consider the total time required to update the data structures.
For fixed $p$ and $i$, when ${\sf OPT}(p, 1, i-1)$ is obtained, the algorithm maintains $D_L(d_{p, i-1}+1, l), D_L(d_{p, i-1}+1, l+1), D_R(l, i-1)$ and $D_R(l+1, i-1)$,
where ${\bm x}^*(1, d_{p, i-1}+1, i-1)$ exists in $[v_l, v_{l+1}]$.
When ${\sf OPT}(p, 1, i)$ is obtained after repeatedly updating these four vertex sets, 
the algorithm maintains $D_L(d_{p, i}+1, l'), D_L(d_{p, i}+1, l'+1), D_R(l', i)$ and $D_R(l'+1, i)$,
where ${\bm x}^*(1, d_{p, i}+1, i)$ exists in $[v_{l'}, v_{l'+1}]$.
Recall that $d_{p, i-1} \le d_{p, i}$ and $l \le l'$ hold by Lemmas \ref{lem1} and \ref{lem2}.
Thus, in order to obtain ${\sf OPT}(p, 1, i)$, the algorithm updates the four vertex sets $2(d_{p, i}-d_{p, i-1}) + 4(l'-l) + 2$ times,
and so, for fixed $p$ and all $i=1, 2, \ldots, n$, the algorithm updates these sets $O(n)$ times in total,
which takes $O(n)$ time by Claim \ref{clm:ds2}.

Therefore, ${\sf OPT}(p, 1, i)$ for all $i=1, 2, \ldots, n$ and $p=1, 2, \ldots, k$ can be obtained in $O(kn)$ time.
\begin{thm}
The minimax $k$-sink location problem in a dynamic path network with uniform capacity can be solved in $O(kn)$ time.
\label{thm:mm1}
\end{thm}

\section{Minisum $k$-sink location problem}
\label{sec:minisum}
In this section, an input graph of this problem is a dynamic path network defined in Section \ref{sec:minimax}.
As a preliminary step, let us consider the minisum 1-sink location problem.

\subsection{Properties of the minisum 1-sink location problem}
\label{sec:minisum1}
Suppose that a sink is located at a point $x \in P$ where $P$ is the input path with $n+1$ vertices.
In continuous model, the cost is defined on each infinitesimal unit of supply,
i.e., the cost of $x$ for a unit is defined as the minimum time required to send the unit to $x$.
Let $sum(x)$ denote the total cost of $x$, i.e., the sum of cost of $x$ for all units on $P$.
Here, let $sum_L(x)$ (resp. $sum_R(x)$) denote the sum of cost of $x$ for all units on $[v_1, x)$ (resp. $(x, v_n]$).
Then, $sum(x)$ is the maximum of $sum_L(x)$ and $sum_R(x)$, i.e., 
\begin{eqnarray}
sum(x) = sum_L(x) + sum_R(x). \label{eq:ms1}
\end{eqnarray}
Without loss of generality, we assume $sum_L(v_1) = 0$ and $sum_R(v_n) = 0$.
Now, suppose that $x$ is located in an open interval $(v_h, v_{h+1})$ with $1 \le h \le n-1$,
then let us explain how function $sum_L(x)$ is determined.

\subsubsection{Case 1:}
For every integer $i$ with $1 \le i \le h$, $\tau(v_i-v_{i-1}) > w_i/c$ holds.
In this case, the first unit of each vertex on $[v_1, v_h]$ can reach $x$ after leaving the original vertex 
without being blocked due to the existence of other units at an intermediate vertex. 
For an integer $i$ with $1 \le i \le h$, let $sum^i(x)$ denote the sum of cost of $x$ for all units of $v_i$.
Here, suppose that there are $\alpha$ units at $v_i$ with sufficiently large $\alpha$, i.e., the size of each unit is equal to $w_i/\alpha$,
and these units continuously reach $x$.
Then by (\ref{eq3.1}), the $l$-th unit finishes reaching $x$ at time $\tau(x-v_i) + l \cdot (w_i/\alpha)/c$.
Therefore, by taking $\alpha$ to the infinity, $sum^i(x)$ can be represented as follows:
\begin{eqnarray}
sum^i(x)  	&=& \lim_{\alpha \to \infty} \sum_{l=1}^\alpha \frac{w_i}{\alpha}\left( \tau(x-v_i) + l \cdot \frac{w_i}{\alpha} \cdot \frac{1}{c} \right) \nonumber \\ 
			&=& \int_0^1 \left( w_i\tau(x-v_i) + \frac{{w_i}^2}{c} \cdot r \right) dr = w_i\tau(x-v_i) + \frac{{w_i}^2}{2c},
\end{eqnarray}
and also $sum_L(x)$ is represented as follows:
\begin{eqnarray}
sum_L(x)  	&=& \sum_{1 \le i \le h} sum^i(x) = \sum_{1 \le i \le h} \left( w_i\tau(x-v_i) + \frac{{w_i}^2}{2c} \right). \label{eq:ms2}
\end{eqnarray}

\subsubsection{Case 2:}
We define the vertex indices $\rho_1, \ldots, \rho_e$ as
\begin{eqnarray}
\begin{array}{ll}
	\rho_1 = \argmax \left\{ \tau(v_h - v_j) + \frac{\sum_{l=1}^j w_l}{c} \ \bigg| \ 1 \le j \le h \right\}				& \mbox{and} \\
	\rho_i = \argmax \left\{ \tau(v_h - v_j) + \frac{\sum_{l=\rho_{i-1}+1}^j w_l}{c} \ \bigg| \ \rho_{i-1} < j \le h \right\}	& \mbox{for} \ 2 \le i \le e.
    \end{array}
 \label{ms2.5}
\end{eqnarray}
Note that $\rho_e = h$ holds. 
For every integer $i$ with $1 \le i \le e$, we also define the value of $\rho_i$ as $\sigma_i = \sum \{ w_h \mid \rho_{i-1}+1 \le h \le \rho_i \}$ where $\rho_0 = 0$.
Here, we notice that for every integer $i$ with $2 \le i \le e$, the first unit of $v_{\rho_{i-1}}$ never be induced to stop at any vertex $v_j$ with $\rho_{i-1} < j < \beta$.
Then, as with (\ref{eq:ms2}), $sum_L(x)$ is represented as follows:
\begin{eqnarray}
sum_L(x)  	&=& \sum_{1 \le i \le e} \left( \sigma_i\tau(x-\rho_i) + \frac{{\sigma_i}^2}{2c} \right). \label{eq:ms3}
\end{eqnarray}
Note that $\sum_{1 \le i \le h^*} \sigma_i = \sum_{1 \le i \le h} w_i$ holds.

\vspace{4mm}

We can compute $sum_R(x)$ in the similar manner as $sum_L(x)$.
Thus, for an open interval $(v_j, v_{j+1})$ with $1 \le j \le n-1$, function $sum(x)$ is linear in $x$ with slope $\tau(\sum_{1 \le i \le j} w_i - \sum_{j+1 \le i \le n} w_i)$.
Now let us consider an open interval $(v_j, v_{j+1})$ with $1 \le j \le n-1$ such that $\sum_{1 \le i \le j} w_i - \sum_{j+1 \le i \le n} w_i \ge 0$ holds.
Then, we can see that for any two points $p, q \in (v_j, v_{j+1})$ with $p < q$, $sum(p) \le sum(q)$ holds.
We will show that for sufficiently small $\epsilon > 0$, $sum(v_j) \le sum(v_j + \epsilon)$ holds.
We confirm
\begin{eqnarray}
sum_R(v_j)	&=&	sum_R(v_j + \epsilon) + \left( \sum_{j+1 \le i \le n} w_i \right) \cdot \tau \epsilon, \ \ \mbox{and} \label{eq:ms4} \\
sum_L(v_j + \epsilon)	&\ge&	sum_L(v_j) + \left( \sum_{1 \le i \le j} w_i \right) \cdot \tau \epsilon. \label{eq:ms5}
\end{eqnarray}
From (\ref{eq:ms4}), (\ref{eq:ms5}) and the assumption of $\sum_{1 \le i \le j} w_i - \sum_{j+1 \le i \le n} w_i \ge 0$, 
we can derive $sum(v_j) \le sum(v_j + \epsilon)$.
In general, we have the following claim.
\begin{clm}
{\rm (i)} For an open interval $(v_j, v_{j+1})$ with $1 \le j \le n-1$ such that $\sum_{1 \le i \le j} w_i - \sum_{j+1 \le i \le n} w_i \ge 0$,
$sum(v_j) \le sum(p)$ holds where $p \in (v_j, v_{j+1})$. \\
{\rm (ii)} For an open interval $(v_j, v_{j+1})$ with $1 \le j \le n-1$ such that $\sum_{1 \le i \le j} w_i - \sum_{j+1 \le i \le n} w_i < 0$,
$sum(v_{j+1}) < sum(p)$ holds where $p \in (v_j, v_{j+1})$. 
\label{clm:ms1}
\end{clm}
Let $x^*$ denote the optimal sink location which minimizes $sum(x)$.
Then, Claim \ref{clm:ms1} implies that $x^*$ is located at some vertex.
\begin{clm}
There exists $x^*$ at a vertex.
\label{clm:ms2}
\end{clm}

\subsection{Algorithm and time complexity for the minisum 1-sink location problem}
We propose the algorithm which can solve the minisum 1-sink location problem in a dynamic path network.
Basically, the algorithm first computes $sum_L(v_i)$ for $2 \le i \le n$ in ascending order of $i$, and next $sum_R(v_i)$ for $1 \le i \le n-1$ in descending order of $i$.
After computing all these values, $sum(v_i)$ can be computed and evaluated for $1 \le i \le n$ in $O(n)$ time.
Then, by Claim \ref{clm:ms2}, the optimal sink location $x^*$ is at a vertex which minimizes $sum(v_i)$ for $1 \le i \le n$.
Below, we show how to compute $sum_L(v_i)$ (computation of $sum_R(v_i)$ can be treated in the similar manner).

First, the algorithm sets $\rho_1 = 1$, $\sigma_1 = w_1$.
By (\ref{eq:ms2}), $sum_L(v_2)$ is computed in $O(1)$ time as follows:
\begin{eqnarray}
sum_L(v_1)	&=&	 \sigma_1\tau(v_2-v_{\rho_1}) + \frac{{\sigma_1}^2}{2c}. \label{eq:ms6}
\end{eqnarray}
Now, suppose that for some integer $j$ with $1 \le j \le n-1$, $h(j)$ has been set as a non-negative integer,
$\rho_i$ and $\sigma_i$ have been obtained for all $i$ with $1 \le i \le h(j)$ in the same manner as mentioned in Case 2, Section \ref{sec:minisum1},
and $sum_L(v_j)$ has been already computed as follows:
\begin{eqnarray}
sum_L(v_j)  	&=& \sum_{1 \le i \le h(j)} \left( \sigma_i\tau(v_j-v_{\rho_i}) + \frac{{\sigma_i}^2}{2c} \right). \label{eq:ms7}
\end{eqnarray}
Let $W_{j-1} = \sum_{1 \le i \le j-1} w_i = \sum_{1 \le i \le h(j)} \sigma_i$ and suppose that $W_{j-1}$ has also been computed.
We then show how to compute $sum_L(v_{j+1})$.
The algorithm newly sets 
\begin{eqnarray}
sum' = sum_L(v_j), \ \ \mbox{and} \ \ W' = W_{j-1}.
\end{eqnarray}
Next, the algorithm tests if $\tau(v_j-v_{\rho_i}) \le w_j/c$ for $1 \le i \le h(j)$ in descending order.
If so, it updates $sum'$ and $W'$ as follows:
\begin{eqnarray}
sum'	\leftarrow sum' - \left( \sigma_i\tau(v_j-v_{\rho_i}) + \frac{{\sigma_i}^2}{2c} \right), \ \ \mbox{and} \ \ W' \leftarrow W' - \sigma_i,
\end{eqnarray}
and deletes $\rho_i$.
If the maximum integer $m$ such that $\tau(v_j-v_{\rho_m}) > w_j/c$ is found or $\tau(v_j-v_{\rho_1}) \le w_j/c$ is obtained, the algorithm stops testing.
In the former case, after the algorithm tests $h(j)-m+1$ times, $\rho_1, \ldots, \rho_{m}$ remain.
Then, after computing $W_j$ as $W_j = W_{j-1} + w_j$,
by (\ref{eq:ms3}), $sum_L(v_{j+1})$ can be computed as
\begin{eqnarray}
sum_L(v_{j+1})	= sum' &+& W'\tau(v_{j+1}-v_j) + \nonumber \\
				&  & \left( (W_j - W')\tau(v_{j+1}-v_j) + \frac{(W_j - W')^2}{2c} \right).
\end{eqnarray}
Also, for the next recursive step, the algorithm eventually sets 
\begin{eqnarray}
h(j+1)=m+1, \ \ \rho_{m+1} = j, \ \ \mbox{and} \ \ \sigma_{m+1} = W_j - W'.
\end{eqnarray}

Since the algorithm tests $h(j)-m+1 = h(j)-h(j+1)+2$ times to compute $sum_L(v_{j+1})$,
it needs to test $\sum_{1 \le i \le n-1} (h(i)-h(i+1)+2)$ times to compute $sum_L(v_i)$ for $2 \le i \le n$.
Here, by $h(1) = 0$, we have
\begin{eqnarray}
\sum_{1 \le i \le n-1} \left( h(i)-h(i+1)+2 \right) = - h(n) + 2(n-1) = O(n).
\end{eqnarray}

\begin{lem}
The minisum 1-sink location problem in a dynamic path network with uniform capacity can be solved in $O(n)$ time.
\label{lem:ms1}
\end{lem}

\subsection{Extension to the minisum $k$-sink location problem}
Let ${\bm x} = (x_1, x_2, \ldots, x_k)$ representing a $k$-sink location given on $P$ and ${\bm d} = (d_1, d_2,$ $\ldots, d_{k-1})$ representing a $(k-1)$-divider given on $P$
(which are defined in the same manner as mentioned in Section \ref{sec:mmp}).
For a given ${\bm d}$, we need only consider ${\bm x}$ such that $x_i $ is given on $[v_{d_{i-1}+1}, v_{d_i}]$ for $1 \le i \le k$, where $d_0 = 0$ and $d_k = n$.
For given ${\bm x}$ and ${\bm d}$,
and for an integer $i$ with $1 \le i \le k$, let $sum_i({\bm x}, {\bm d})$ denote the sum of cost of $x_i$ for all supplies on $[v_{d_{i-1}+1}, v_{d_i}]$.
Letting $sum({\bm x}, {\bm d}) = \sum \{ sum_i({\bm x}, {\bm d}) \mid 1 \le i \le k \}$, the minisum $k$-sink location problem is defined as follows:
\begin{eqnarray}
{\rm Q_{minisum}}(P): \ {\rm minimize} \ \left\{ sum({\bm x}, {\bm d}) \ \bigg| \ {\bm x} \in P^k \ {\rm and} \ {\bm d} \in \{ 1, 2, \ldots, n \}^{k-1} \right\}.
\label{pro2}
\end{eqnarray}

We below show that this problem can be transformed to an equivalent problem, which requires to find the minimum $k$-link path in a weighted, complete, directed acyclic graph (DAG) \cite{s98}.
First, for integers $i$ and $j$ with $1 \le i < j \le n+1$, let ${\sf OPT}(i, j)$ denote the optimal cost for the minisum $1$-sink location problem in $[v_i, v_{j-1}]$.
Let us consider a DAG $G = (N, A)$ such that $N = \{ u_1, u_2, \ldots, u_n, u_{n+1} \}$ and
for every vertex pair $(u_i, u_j)$ with $1 \le i < j \le n+1$, 
there exists an edge which is directed from $u_i$ to $u_j$ and associated with the weight of ${\sf OPT}(i, j)$.
Then, ${\rm Q_{minisum}}(P)$ is equivalent to a problem requiring to find a path in $G$ from $u_1$ to $u_{n+1}$ which contains exactly $k$ edges such that the sum of weights is minimized.
Schieber \cite{s98} showed that this problem can be solved by querying edge weights $O(n \cdot \min \{ k, 2^{\sqrt{\log k \log \log n}}\})$ times
if the input DAG satisfies the {\it concave Monge property}, that is, ${\sf OPT}(i, j) + {\sf OPT}(i+1, j+1) \le {\sf OPT}(i+1, j) + {\sf OPT}(i, j+1)$ holds for any integers $i$ and $j$ with $1 \le i+1 < j \le n$.
Since each weight query takes $O(n)$ time by Lemma \ref{lem:ms1}, if the concave Monge property is proved, ${\rm Q_{minisum}}(P)$ can be solved in $O(n^2 \cdot \min \{ k, 2^{\sqrt{\log k \log \log n}}\})$ time.
Therefore, we prove the following lemma.
\begin{lem}
For any integers $i$ and $j$ with $1 \le i+1 < j \le n$,
${\sf OPT}(i, j) + {\sf OPT}(i+1, j+1) \le {\sf OPT}(i+1, j) + {\sf OPT}(i, j+1)$ holds.
\label{lem:ms2}
\end{lem}
\begin{proof}
For integers $i$ and $j$ with $1 \le i < j \le n+1$, and a 1-sink location $x \in [v_i, v_{j-1}]$,
let $sum_{i, j}(x)$ denote the sum of cost of $x$ for all supplies on $[v_i, v_{j-1}]$,
and let $sum_L^i(x)$ (resp. $sum_R^j(x)$) denote the sum of cost of $x$ for all supplies on $[v_i, x)$ (resp. $(x, v_{j-1}]$).
Also, let $x^*(i, j) = \argmin \{ sum_{i, j}(x) \mid x \in [v_i, v_{j-1}] \}$.
By the definitions, we have
\begin{eqnarray}
sum_{i, j}(x) 	&=&	 sum_L^i(x) + sum_R^j(x), \ \mbox{and} \label{eq:ms2.1} \\
{\sf OPT}(i, j) 	&=&	 sum_{i, j}(x^*(i, j)). \label{eq:ms2.2}
\end{eqnarray}
Then, we consider two cases: [Case 1] $x^*(i+1, j) \le x^*(i, j+1)$ and [Case 2] $x^*(i+1, j) > x^*(i, j+1)$.
Here, let us prove only Case 1 (Case 2 can be symmetrically proved).
We first show that 
\begin{eqnarray}
sum_{i, j}(x^*(i+1, j)) &-&sum_{i+1, j}(x^*(i+1, j)) \nonumber \\
&\le& sum_{i, j+1}(x^*(i, j+1)) - sum_{i+1, j+1}(x^*(i, j+1)).
\label{eq:ms2.3}
\end{eqnarray}
By (\ref{eq:ms2.1}), the left side of (\ref{eq:ms2.3}) is equal to $sum_L^i(x^*(i+1, j)) - sum_L^{i+1}(x^*(i+1, j))$ and the right side of (\ref{eq:ms2.3}) is equal to $sum_L^i(x^*(i, j+1)) - sum_L^{i+1}(x^*(i, j+1))$.
Let $D = sum_L^{i+1}(x^*(i, j+1)) - sum_L^{i+1}(x^*(i+1, j))$ (clearly $D > 0$), that is,
\begin{eqnarray}
sum_L^{i+1}(x^*(i, j+1)) = sum_L^{i+1}(x^*(i+1, j)) + D.
\label{eq:ms2.4}
\end{eqnarray}
Then, we have 
\begin{eqnarray}
sum_L^i(x^*(i, j+1)) \ge sum_L^i(x^*(i+1, j)) + D.
\label{eq:ms2.5}
\end{eqnarray}
By (\ref{eq:ms2.4}) and (\ref{eq:ms2.5}), we obtain
\begin{eqnarray}
sum_L^i(x^*(i+1, j)) &-& sum_L^{i+1}(x^*(i+1, j)) \nonumber \\
&\le& sum_L^i(x^*(i, j+1)) - sum_L^{i+1}(x^*(i, j+1)),
\label{eq:ms2.6}
\end{eqnarray}
which is equivalent to (\ref{eq:ms2.3}) as mentioned above.
On the other hand, by the optimality of ${\sf OPT}(i, j)$ and ${\sf OPT}(i+1, j+1)$, we have
\begin{eqnarray}
sum_{i, j}(x^*(i+1, j)) 		&\ge& {\sf OPT}(i, j), \ \mbox{and} \label{eq:ms2.7} \\
sum_{i+1, j+1}(x^*(i, j+1))	&\ge& {\sf OPT}(i+1, j+1). \label{eq:ms2.8}
\end{eqnarray}
Then, by (\ref{eq:ms2.3}), (\ref{eq:ms2.7}), (\ref{eq:ms2.8}) and the definitions of $sum_{i+1, j}(x^*(i+1, j)) = {\sf OPT}(i+1, j)$ and $sum_{i, j+1}(x^*(i, j+1)) = {\sf OPT}(i, j+1)$, we obtain
\begin{eqnarray}
{\sf OPT}(i, j) - {\sf OPT}(i+1, j) \le {\sf OPT}(i, j+1) - {\sf OPT}(i+1, j+1),
\end{eqnarray}
which implies that the lemma holds in Case 1.
\qed
\end{proof}

\begin{thm}
The minisum $k$-sink location problem in a dynamic path network with uniform capacity can be solved in $O(n^2 \cdot \min \{ k, 2^{\sqrt{\log k \log \log n}}\})$ time.
\label{thm:ms1}
\end{thm}

\section{Conclusion}
In this paper, we study the $k$-sink location problem in dynamic path networks with continuous model assuming that edge capacity is uniform and sinks can be located at any point in the network, and prove that the minimax problem can be solved in $O(kn)$ time and the minisum problem can be solved in $O(n^2 \cdot \min \{ k, 2^{\sqrt{\log k \log \log n}}\})$ time.

On the other hand, we leave as an open problem to reduce the time bound to $O(kn)$ for the minisum problem,
and extend the solvable networks into dynamic path networks with general capacities or more general networks (e.g., trees).

\end{document}